\documentclass[journal, onecolumn, 12pt,draftclsnofoot,]{IEEEtran}
\usepackage{times}
\usepackage{lipsum}
\usepackage{mwe}
\usepackage{fancybox}
\usepackage{multicol}
\usepackage{color}
\usepackage{graphicx}
\usepackage{epsfig}
\usepackage{graphicx,latexsym}
\usepackage{amssymb,amsthm,amsmath}
\usepackage{array}
\usepackage{xy}
\usepackage{amsthm}
\usepackage{mathtools}
\usepackage{epstopdf}
\usepackage{algorithm}
\usepackage{algpseudocode}
\usepackage{cite}
\usepackage{ amssymb }
\usepackage[table,xcdraw]{xcolor}
\usepackage{float}
\usepackage{multirow}

\usepackage{hhline}
\usepackage{url}
\usepackage{booktabs}
\usepackage{amsmath,algorithm,tabularx}
\usepackage{subfigure}

\renewcommand{\theenumi}{\arabic{enumi}}

\newcolumntype{C}[1]{>{\centering\arraybackslash$}p{#1}<{$}}
\makeatletter

\newcommand{\Rmnum}[1]{\expandafter\@slowromancap\romannumeral #1@}

\newtheorem{theorem}{Theorem}

\newtheorem{lemma}[theorem]{Lemma}

\newcommand{\multiline}[1]{%
  \begin{tabularx}{\dimexpr\linewidth-\ALG@thistlm}[t]{@{}X@{}}
    #1
  \end{tabularx}
}
\makeatother
\begin{document}
\title{{Deep Reinforcement Learning for Collaborative Edge Computing in Vehicular Networks}}
\author{Mushu~Li,~\IEEEmembership{Student~Member,~IEEE,}
        Jie~Gao,~\IEEEmembership{Member,~IEEE,}
        Lian~Zhao,~\IEEEmembership{Senior~Member,~IEEE,}
        and Xuemin~(Sherman)~Shen,~\IEEEmembership{Fellow,~IEEE} \thanks{
        
        Mushu Li, Jie Gao, and Xuemin (Sherman) Shen are with the Department of Electrical and Computer Engineering, University of Waterloo, Waterloo, ON, Canada, N2L 3G1 (email:\{m475li, jie.gao, sshen\}@uwaterloo.ca).

Lian Zhao is with the Department of Electrical, Computer and Biomedical Engineering, Ryerson University, Toronto, ON, Canada, M5B 2K3 (email:l5zhao@ryerson.ca).}

}%

\maketitle
\thispagestyle{empty}
\begin{abstract}
Mobile edge computing (MEC) is a promising technology to support mission-critical vehicular applications, such as intelligent path planning and safety applications. 
In this paper,
a collaborative edge computing framework is developed to reduce the computing service latency and improve service reliability for vehicular networks.
First, a task partition and scheduling algorithm (TPSA) is proposed to decide the workload allocation and schedule the execution order of the tasks offloaded to the edge servers given a computation offloading strategy. Second, an artificial intelligence (AI) based collaborative computing approach is developed to determine the task offloading, computing, and result delivery policy for vehicles. Specifically, the offloading and computing problem is formulated as a Markov decision process. A deep reinforcement learning technique, \textit{i.e.}, deep deterministic policy gradient, is adopted to find
the optimal solution in a complex urban transportation network. By our approach, the service cost,
which includes computing service latency and service failure penalty, can be minimized via the optimal workload assignment and server selection in collaborative computing. Simulation results show that the proposed AI-based collaborative computing approach can adapt to a highly dynamic environment with outstanding performance.

\end{abstract}

\section{Introduction}

Vehicular communication networks have drawn significant attention from both academia and industry in the past decade. Conventional vehicular networks aim to improve the driving experience and enable safety applications via data exchange in vehicle-to-everything (V2X) communications. 
In the era of 5G, the concept of vehicular networks has been extended to Internet-of-Vehicle (IoV), in which intelligent and interactive applications are enabled by communication and computation technologies \cite{Cheng_v1}. A myriad of on-board applications can be implemented in the context of IoV, such as assisted/autonomous driving and platooning, urban traffic management, and on-board infotainment services~\cite{Peng_v1,Gao}.

Although IoV technologies are promising, realizing the IoV applications still faces challenges. One of the obstacles is the limited on-board computation capability at vehicles. For example, a self-driving car with ten high-resolution cameras may generate 2 gigapixels per second of data, while 250 trillion computation operations per second are required to process the data promptly \cite{Nvidia}.
Processing such computation-intensive applications on vehicular terminals is energy-inefficient and time-consuming. To overcome the limitation, mobile edge computing (MEC) is an emerging paradigm that provides fast and energy-efficient computing services for vehicle users \cite{Zhang2, Liu, Ning}. 
Via vehicle-to-infrastructure (V2I) communications, resource-constrained vehicle users are allowed to offload their computation-intensive tasks to highly capable edge servers co-located with roadside units (RSUs) for processing. 
Meanwhile, compared to the conventional mobile cloud computing, the network delay caused by task offloading can be significantly reduced in MEC due to the proximity of the edge server to vehicles \cite{cong1}. 
Consequently, some applications that require high computing capability, such as path navigation, video stream analytics, and objective detection, can be implemented in vehicular networks with edge servers \cite{Zhang_proc}.

Despite the advantage brought by MEC-enabled vehicular networks, new challenges have emerged in task offloading and computing. One critical problem in MEC is to decide which edge servers should their computing tasks be offloaded to. In vehicular networks, the highly dynamic communication topology leads to unreliable communication links \cite{Lyu1}. Due to the non-negligible computing time and the limited communication range of vehicles, a vehicle may travel out of the coverage area of an edge server during a service session, resulting in a service disruption. To support reliable computing services for  high-mobility users, a service migration scheme has been introduced in~\cite{Taleb}. Under this scope, when a user moves out of the communication area of the edge that the computing task was offloaded, the computing process will be interrupted, and the corresponding virtual machine (VM) will be migrated to a new edge according to the radio association.
%, while the cost of service interruption and network overhead for virtual machine migration cannot be avoided \cite{Farhadi}. 
In the urban area, where highly dense infrastructure are deployed, frequent service interruption would happen due to the dynamically changing radio association, which can significantly increase the overall computing service latency.

Alternatively, computing service reliability can be achieved by cooperation among edge servers. 
{Different from service migration, which achieves service reliability by migrating the computing service according to the vehicle's trajectory, service cooperation improves the service reliability by accelerating task processing time. The computing task can be divided and computed by multiple servers in parallel or fully offloaded to a server with high computing capability at the cost of communication overhead \cite{Cao,Lin}. In this regard, the computing task can be forwarded to the edge server which is out of the user's communication range. Compared to service migration, in which edge servers only execute the task offloaded by the vehicles under their communication coverage, service cooperation allows edge servers processing the tasks offloaded by the vehicles out of their coverage for reducing the overall computing time. }
%does not require a direct communication link between the edge and vehicles
%in service cooperation, edge servers 
%may execute the computing tasks offloaded by vehicles which is .
%Compared to service migration, in which the computation association of a vehicle is consistent with its communication association, in service cooperation, the radio association of the vehicle and the computation association of its computing task are evaluated separately.}
Nevertheless, multi-hop communications could result in significant transmission delay and waste communication spectrum resources in the task offloading process. The tradeoff between the communication overhead and the computing capability increases the complexity of the server assignment problem.  In addition, although computing service latency can be reduced by cooperative computing, it is hard to guarantee  service reliability for the vehicles with high mobility. The uncertainty of vehicle moving trajectories poses significant challenges in computing result delivery.

Motivated by the issues in the existing service migration and computing cooperation schemes, we present a computing collaboration framework to provide reliable low-latency computing in an MEC-enabled vehicular network. Once an edge server receives the computing tasks offloaded by a vehicle, it may partially or fully distribute the computing workload to another edge server to reduce computing latency. Furthermore, by selecting proper edge servers to deliver the computing results, vehicle users are able to obtain computing results without service disruption caused by mobility. 
Under this framework, we propose a novel task offloading and computing approach that reduces the overall computing service latency and improves service reliability. To achieve this objective, we firstly formulate a task partition and scheduling optimization problem, which allows all received tasks in the network to be executed with minimized latency given the offloading strategy. A heuristic task partition and scheduling approach is developed to obtain a near-optimal solution of the non-convex integer problem. {In addition, we formulate the radio and computing association problem into a Markov decision process (MDP). By characterizing stochastic state transitions in the network, MDP is able to provide proactive offloading policy for vehicles. An artificial intelligence (AI) approach, deep reinforcement learning (DRL), is adopted to cope with the curse of dimensionality in MDP and unknown network state transitions caused by vehicle mobility.} Specifically, a convolutional neural network (CNN) based DRL is developed to handle the high-dimensional state space, and the deep deterministic policy gradient (DDPG) algorithm is adopted to handle the high-dimensional action space in the proposed problem.
The major contributions of this paper are:

\begin{enumerate}
\item We develop an efficient collaborative computing framework for MEC-enabled vehicular networks to provide low-latency and reliable computing services. 
To overcome the complexity brought by the dynamic network topology, we propose a location-aware task offloading and computing strategy to guide MEC server collaboration.
%Compared to the existing approaches, our method address the computing reliability during MEC server selection while reducing overall service time.
\item We devise a task partition and scheduling scheme to divide the computing workload among edge servers and coordinate the execution order for tasks offloaded to the servers. Given the offloading strategy, our scheme can minimize the computing time by finding a near-optimal task scheduling solution with low time-complexity.

\item We further propose an AI-based collaborative computing approach, which utilizes a model-free method to find the optimal offloading strategy and MEC server assignment in a 2-dimensional transportation system. A CNN based DDPG technique is developed to capture the correlation of the state and action among different zones and accelerate the learning speed. 
\end{enumerate}

The remainder of the paper is organized as follows. In Section II, we present the related works. Section~III describes the system model. Section IV formulates the service delay minimization problem. In Section V, we present the task partition and scheduling scheme, followed by an AI-based collaborative computing approach in Section VI. Section VII presents simulation results, and Section VIII concludes the paper.

\section{Related Works}
\subsection{Mobile Edge Computing}
As proposed by ETSI in \cite{WP}, the main objective of MEC is to reduce the computing task offloading and computing latency via utilizing the computing resources located in edge devices, such as base stations and access points. In the context of edge computing, one of the main problems is to determine the computing task offloading mechanisms. The edge server selection problem has been evaluated in \cite{Cheng} and \cite{Liu2}. In \cite{Cheng}, Cheng \textit{et al.} propose a user association strategy to jointly minimize the computing delay, user energy consumption, and the server computing cost under a space-air-ground integrated network. 
A model-free approach is proposed in the work to deal with the complex offloading decision-making problem.  In \cite{Liu}, Liu \textit{et al.} investigate the user-server association policy, which takes into account the communication link quality and server computing capability. 
%The main objective is to minimize the computing latency while stabilizing the processing queue in the edge servers. 
In both works, the computing association follows the radio association, \textit{i.e.}, the computing task is processed within the edge server that the task is offloaded. 
{To further reduce the computation time, task partition has been considered in \cite{Chen2,You,Li}. Computing tasks can be split and computed by multiple servers in parallel. The cooperation computing has been investigated among the works \cite{Chen2,You,Li} under different network environments, while the impact of user mobility has not been addressed.  
Additionally, in \cite{Cao, Alameddine}, and \cite{Feng_v1}, task scheduling, \textit{i.e.}, ordering the task execution sequences, is also evaluated in the offloading decision making process. In those works, the task scheduling problem is formulated into a mixed-integer programming problem, and heuristic algorithms are proposed to obtain near-optimal solutions efficiently.} Different from the above works, we investigate the task partition and scheduling under the collaborative computing framework, in which the adjustment on workload allocation for a task can affect the performance of other tasks, which makes the problem more complex.

\subsection{MEC-enabled Vehicular Networks}
The problem of computing offloading has been investigated in many research works in the context of vehicular networks \cite{He,Qi,Wang,Li3}. In those works, the main objective is to minimize service time by selecting the optimal edge server, while service reliability in the presence of vehicle mobility is not taken into account. 
In \cite{Sun,Ning2,Hafeez}, machine learning techniques are adopted to obtain the reliable offloading decision for vehicles via predict the vehicle trajectories. In \cite{Sun}, Sun \textit{et al.} focus on task offloading and execution utilizing the computing resources on vehicles, \textit{i.e.} vehicular edge. An online learning algorithm, \textit{i.e.}, multi-armed bandit, is utilized to determine the computing and communication association among vehicles. In \cite{Ning2}, Ning \textit{et al.} apply a DRL approach to jointly allocate the communication, caching, and computing resources in the dynamic vehicular network.
Furthermore, to deal with service disruption when the vehicle leaving the server converge, service migration has been firstly proposed in \cite{Taleb}. According to the vehicle moving trajectory, the corresponding computing services can be migrated to another edge server that may associate the vehicle in the future. 
The proactive service migration strategy has been investigated in \cite{Wang2} and \cite{Liang}, where MDP is utilized to make the migration decision in a proactive manner.  
{To alleviate service interruption and network overheads in virtual machine migration, server cooperation has been studied in \cite{Zhang2, Zhou_v1}, and \cite{Xiao_v1}. The works \cite{Zhang2} and \cite{Zhou_v1} consider that vehicles divide and offload the computing tasks to multiple servers according to the predicted traveling traces. Vehicle-to-vehicle communication is used to disseminate the computing result if the edge server cannot connect with the vehicle at the end of a service session. In \cite{Xiao_v1}, the work utilizes neural networks to predict the computing demand in the vehicular network. MEC servers are clustered to compute the offloaded tasks cooperatively. } 
Different from the above works, our proposed approach achieves service reliability improvement by collaboration and task scheduling among edge servers without cooperative transmission, which reduces the communication overhead of result delivery.

%Our previous work \cite{Li2} proposes a model-free approach to determine task offloading and computing collaboration strategy under a highway scenario. Motivated by the work, in this paper, we further consider an urban scenario, in which a 2-dimension transportation network is evaluated such that the network becomes more complex. We also consider the task partition and scheduling scheme to further reduce the computing service latency and improve service reliability.

\section{System Model}
\subsection{Collaborative Edge Computing Framework}

\begin{figure*}[t]  
  \centering  
  \includegraphics[width=150mm]{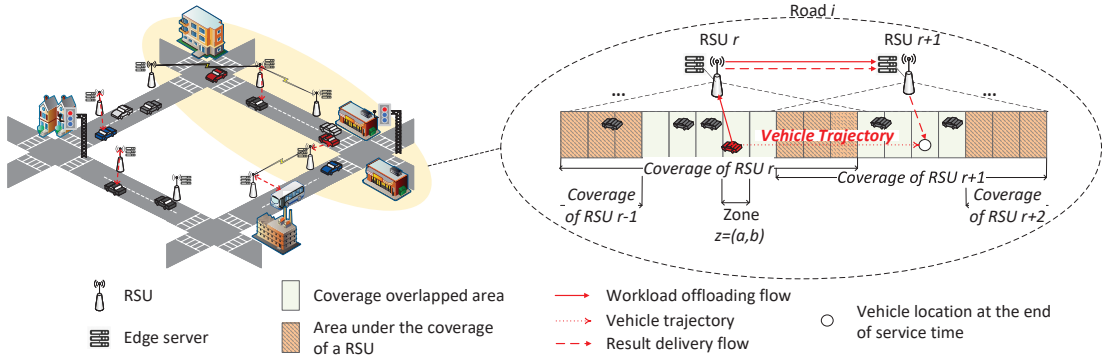}\\
  \caption{Network model. }  
  \label{fig:sm} 
\end{figure*}
An MEC-enabled vehicular network is illustrated in Fig. \ref{fig:sm}. 
A row of RSUs, equipped with computing resources, provide seamless communication and computing service coverage for vehicles on the road. 
An RSU can also communicate with other RSUs within its communication range via wireless links. The set of RSUs is denoted by $\mathcal{R}$, where the index of RSUs is denoted by $r\in \mathcal{R}$. 
%Considering 2-dimension Cartesian coordinate system, the location of RSU $i$ is denoted by $(x_r, y_r)$. 
We assume that a global controller has full knowledge of the transportation network and makes offloading and computing decisions for all the vehicles in a centralized manner. In our model, a computing session for a task includes three steps:
\begin{enumerate}
\item Offloading: When a computing task is generated at a vehicle, the vehicle selects an RSU, which is under its communication range, and offloads the computing data of the task to the RSU immediately. In the example shown in Fig. \ref{fig:sm}, RSU $r$ is selected to offload the computing load. Such RSU is referred to as the \textit{receiver RSU} for the task. 
\item Computing: After the computing task is fully offloaded, the receiver RSU can process the whole computing task or select another RSU to share the computing load.  
The RSU, which is selected to process the task collaboratively with the receiver RSU, is referred to as the \textit{helper RSU} for the task. 
\item Delivering: A vehicle may travel out of the communication range of its receiver RSU. Therefore, the controller may select an RSU, which could connect with the vehicle at the end of service session, to gather and transmit computing results. The RSU is referred to as the \textit{deliver RSU}.
To reduce the overhead, we limit the deliver RSU to be either the receiver RSU or the helper RSU of the task. In the example shown in Fig. \ref{fig:sm}, RSU $r+1$ behaves as both the helper RSU and the deliver RSU for the computing task offloaded by the vehicle.
\end{enumerate}

To reduce the decision space in task offloading and scheduling, instead of providing the offloading and computing policy to
individual vehicles, we consider location-based offloading and
computing policy. We divide each road into several zones with equal length, where the set of zones is denoted by $\mathcal{Z}$.
The index of the zones is denoted by $z = (a,b) \in \mathcal{Z}$. The terms $a$ and $b$ represent the index of the roads and the index of the segments on the road, respectively, where $a\in \{1, \dots, A\}$, and $b\in \{1, \dots, B\}$.  As the vehicle drives through the road, it traverses the zones consecutively. We assume that all vehicles in the same zone follow the same offloading and computing policy.\footnote{The accuracy of vehicle locations will be improved when the length of the zone is reduced. In consideration of the length of a car, the length of a zone is larger than 5 m.} For simplicity, we evaluate the aggregated tasks for vehicles in each zone at a time slot, and refer to the tasks offloaded by zone $z$ as task $z$ in the remainder of the paper.
We suppose that the vehicle will not travel out of a zone during the time duration of a time slot, and vehicles can complete the offloading process of a task generated in a zone before it travels out of the zone.
Denote the set of vehicles in zone $z$ and time slot $t\in \mathcal{T}$ as $\mathcal{V}_{z,t}$. The offloading decision for vehicles in zone $z$ and time slot $t$ is represented by a
vector $\boldsymbol{\alpha}_{z,t} \in \mathbb{Z}_+^{|\mathcal{R}|}$, where $\sum_{r = 1}^{|\mathcal{R}|} \alpha_{z,r,t} = 1$.  The element $\alpha_{z,r,t}$ is 1 if RSU $r$ is selected as the receiver RSU for the vehicles in zone $z$ and time slot $t$, and 0 otherwise. Similarly, the collaborative computing decision for vehicles in zone $z$ and time slot $t$ is represented by a vector $\boldsymbol{\beta}_{z,t} \in \mathbb{Z}_+^{|\mathcal{R}|}$, where $\sum_{r = 1}^{|\mathcal{R}|} \beta_{z,r,t} = 1$. The element $\beta_{z,r,t}$ is 1 if RSU $r$ is selected as the helper RSU for the vehicles in zone $z$ and time slot $t$, and 0 otherwise. In addition, the decision on result delivery is denoted by a binary variable $\gamma_{z,r,t}$, where $\gamma_{z,r,t}$ is 0 if the computing results are delivered by RSU $r$ for task $z$ in time slot $t$, and $\gamma_{z,r,t}$ is 1 if the computing results are delivered by RSU $r$.

\subsection{Cost Model}
In this paper, the system cost includes two parts: the service delay and the penalty caused by service failure.
\subsubsection{Service Delay}

\begin{figure}[t]  
  \centering  
  \includegraphics[width=90mm]{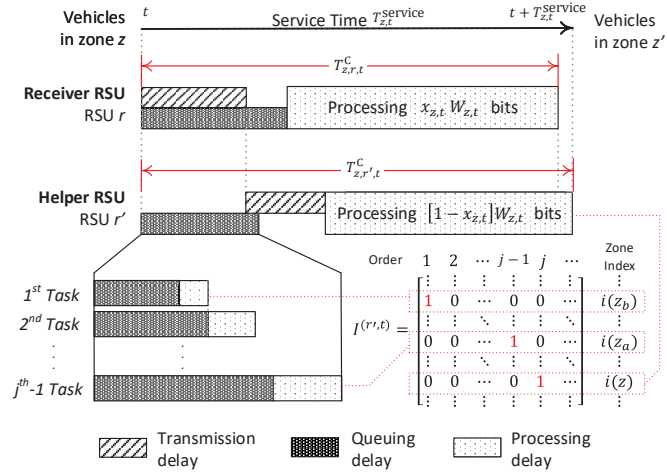}\\
  \caption{An example of the task offloading and computing process. }  
  \label{fig:sm2} 
\end{figure}

We adopt the task partition technique during task processing. Once a receiver RSU receives the offloaded task from vehicles in a zone, it immediately divides the task and offloads a part of the workload to the helper RSU of the corresponding zone. We denote the computing delay of task $z$ corresponding to the receiver or helper RSU $r$ in time slot $t$ as $T^{\textrm{C}}_{z,r,t}$. As shown in Fig. \ref{fig:sm2}, the computing delay includes task offloading delay, queuing delay, and processing delay. Since the amount of output data is usually much smaller compared to the amount of input data, we neglect the transmission delay in result delivery \cite{Wang,Li}.

Firstly, task offloading comprises two steps: offloading tasks from vehicles to their receiver RSU and offloading the partial workload from the receiver RSU to the helper RSU. 
According to the propagation model in 3GPP standards \cite{Chen}, the path loss between a transmitter and a receiver with distance $d$ (km) can be computed as:
\begin{align}
&L(d) =  40 (1-4\times 10^{-3}D^{hb})\log_{10}d -  18\log_{10}D^{hb}\\ \notag
&\hspace{1cm} +21\log_{10}f + 80 \textrm{ (dB)},
\end{align}
where the parameter $f$ is the carrier frequency in MHz, and the parameter $D^{hb}$ represents the antenna height in meter. We do not consider the shadowing effect of the channel.
%$X$ represents the shadowing fading of the channel, which follows the log-normal distribution.
Denote the distance between the center point of zone $z$ and the location of RSU $r$ as $D_{z,r}$, and the distance between RSU $r$ and $r'$ as $D_{r,r'}$. The data rate for vehicles in zone $z$ offloading task to RSU $r$ is
\begin{equation}
r_{z,r} = B^{\textrm{Z}} \log_2 \Big(1+\frac{P^\textrm{V}10^{-L(D_{z,r})/10}}{\sigma_v^2}\Big),
\end{equation}
where the parameter {$\sigma_v^2$} denotes the power of the Gaussian noise in the V2I channel, $P^\textrm{V}$ represents the vehicle transmit power, and $B^\textrm{Z}$ represents the bandwidth reserved for vehicles in a zone. As the  receiver RSU for task $z$, a signal-to-noise ratio threshold should be satisfied, where
\begin{align}
\frac{P^V10^{-L(D_{z,r})/10}}{\sigma_v^2}&\geq \alpha_{z,r,t}\delta^{\textrm{O}}, \forall t, z, r, \label{c1}
\end{align}
where $\delta^{\textrm{O}}$ is the signal-to-noise ratio threshold for data offloading.
Assume that vehicles in a zone are scheduled to offload the tasks successively, and the channel condition is fixed in the duration of any computing task offloading. 
The transmission delay for offloading the computing data in zone $z$ to the receiver RSU is:
\begin{equation}
T^{\textrm{T}}_{z,t} = \sum_{r \in \mathcal{R}} \frac{\alpha_{z,r,t}W_{z,t}}{r_{z,r}},
\end{equation} 
where $W_{z,t}$ represents the overall computing data generated by vehicles in zone $z$, \textit{i.e.}, task $z$, and time slot $t$.
In addition, the data rate between RSU $r$ and RSU $r'$ for forwarding the computing data offloaded from a zone is
\begin{equation}
r_{r,r'} = B^\textrm{R} \log_2\Big(1+\frac{P^\textrm{R} 10^{-L(D_{r,r'})/10}}{\sigma_r^2}\Big),
\end{equation} 
where the parameter {$\sigma_r^2$} represents the power of the Gaussian noise in the RSU to RSU channel, $P^\textrm{R}$ represents the RSU transmit power, and $B^\textrm{R}$ represents the bandwidth reserved for forwarding data offloaded from a zone. In data forwarding, the signal-to-noise constraint is also required to be satisfied, where
\begin{align}
\frac{P^R 10^{-L(D_{r,r'})/10}}{\sigma_r^2}&\geq \beta_{z,r',t}\delta^{\textrm{O}}, \forall t, z, r, r'. \label{c2}
\end{align}
For computing task $z$ in time slot $t$, the portion of workload to be processed by the receiver RSU and the helper RSU is denoted by $x_{z,t}$ and $1-x_{z,t}$, respectively. Thus, the delay for forwarding the data to the deliver RSU is:
\begin{equation}
T^{\textrm{R}}_{z,t} = \sum_{r\in \mathcal{R}} \sum_{r'\in \mathcal{R}}\frac{\alpha_{z,r,t}\beta_{z,r',t}(1-x_{z,t})W_{z,t}}{r_{r,r'}}.
\end{equation}

Furthermore, after the task is offloaded to edge servers, the queuing delay may be experienced. Let set $\mathcal{Z}^{r,t}$ denote the zones which have tasks offloaded to RSU $r$, \textit{i.e.}, $\{z|\alpha_{z,r,t} = 1\} \cup \{z|\beta_{z,r,t} = 1\}$, and let $i(z)$ represent the index of zone $z$ in set $\mathcal{Z}^{r,t}$. We denote $N_{r,t}$ as the number of tasks offloaded in time slot $t$ and assigned to the RSU $r$, where $N_{r,t} = \sum_z \alpha_{z,r,t}+\beta_{z,r,t}$. 
Then, a matrix, $\mathbb{I}^{(r,t)} \in \mathbb{Z}_+^{N_{r,t} \times N_{r,t}}$, can be defined to imply the processing order of tasks offloaded to RSU $r$ in time slot $t$, where $I^{(r,t)}_{i(z),j} = 1$ if the task offloaded from zone $z$ is scheduled as the $j$-th task to be processed among the other tasks offloaded in the same time slot. As shown in Fig.~\ref{fig:sm2}, the queuing delay of a task depends on the computing time of the task scheduled priorly.  
For the first task to be processed among the tasks offloaded in time slot $t$, the queuing delay stems from the computing time for the tasks offloaded in previous time slots. Thus, the queuing delay of task $z$ in RSU $r$ can be formulated as follows:
\begin{equation}
T^{\textrm{Q}}_{z,r,t} = 
\begin{cases} 
T^{\textrm{Q0}}_{r,t}, \textrm{ if } I^{(r,t)}_{i(z),1} = 1,  \\
\sum_{z'} \sum_j I^{(r,t)}_{i(z),j} I^{(r,t)}_{i(z'),j-1} T^{\textrm{C}}_{z',r,t}, \textrm{ otherwise.} 
\end{cases}
\end{equation}
The term $T^{\textrm{Q0}}_{r,t}$ represents the latency for finishing the tasks offloaded in previous time slots $\{1, \dots, t-1\}$, where 
\begin{equation}
T^{\textrm{Q0}}_{r,t} = \max \big\{ \sum_{z'} I^{(r,t)}_{i(z'),N_{r,t-1}}T^{\textrm{C}}_{z',r,t-1} - \epsilon, 0 \big\},
\end{equation}
where $\epsilon$ is the length of a time slot.

We consider that data transmission and task processing run in parallel. After the task is offloaded and other tasks scheduled priorly are completed, the task can be processed by the dedicated server. The delay for processing task $z$ offloaded to RSU $r$ in time slot $t$ can be formulated as
\begin{equation}
T^{\textrm{P}}_{z,r,t} = \frac{\chi W_{z,t}[\alpha_{z,r,t}x_{z,t}+\beta_{z,r,t}(1-x_{z,t})]}{C_r},
\end{equation} 
where $C_r$ denotes the computing capability (CPU-cycle frequency) of RSU $r$, and $\chi$ denotes the number of computation cycles needed to execute 1 bit of data. 

Given the offloading delay, queuing delay, and processing delay, the computing delay for task $z$ on RSU $r$ can be formulated as follows:
\begin{equation}
T^{\textrm{C}}_{z,r,t} = \max\{T^{\textrm{T}}_{z,t}+ \beta_{z,r,t}T^{\textrm{R}}_{z,t}, T^{\textrm{Q}}_{z,r,t} \} + T^{\textrm{P}}_{z,r,t}.
\label{eq.ser}
\end{equation}
Denote the overall service delay for the task offloaded from zone $z$ in time slot $t$ as $T^{\textrm{service}}_{z,t}$. As shown in Fig. \ref{fig:sm2}, the overall service delay depends on the longest computing time between the receiver RSU and the helper RSU.  
Thus, we have
\begin{equation}
T^{\textrm{service}}_{z,t}=\max \{\sum_r \alpha_{z,r,t} T^{\textrm{C}}_{z,r,t}, \sum_r \beta_{z,r,t}T^{\textrm{C}}_{z,r,t}\}.
\label{eq.service}
\end{equation}

\subsubsection{Service Failure Penalty}
The mobility of vehicles brings uncertainty in result downloading. Service failure may occur if a vehicle is out of the coverage of its deliver RSU during the service session. 
{Denote the zone that vehicle $v$ is located when its computing result is delivered as $m_v$, \textit{i.e.}, the location of vehicle $v\in \mathcal{V}_{z,t}$ in time slot $t+ T_{z,t}^{\textrm{Service}}$.}
Also, we denote the signal-to-noise ratio threshold for result delivering as $\delta^{\textrm{D}}$. We introduce a variable $\textbf{1}_{z,t}$ to indicate whether the computing service for task $z$ offloaded in time slot $t$ is successful or not, where
\begin{equation}
\textbf{1}_{z,t} = 
\begin{cases}
1, \textrm{ if }{P^R10^{-L(D_{m_v,r})/10}} \geq {\sigma_r^2}\gamma_{z,r,t}\delta^{\textrm{D}}, \forall v \in \mathcal{V}_{z,t}\\
0, \textrm{ otherwise. }
\end{cases}
\end{equation}
%where $w$ represents the index of deliver RSU for the tasks offloaded in zone $z$ and time slot $t$, \textit{i.e.}, $\gamma_{z,w,t} = 1$.

\section{Problem Formulation}
Our objective is to minimize the weighted sum of the overall computing service delay for vehicle users and service failure penalty. The corresponding objective function can be formulated as follows:

\begin{subequations}
\label{problem}
\begin{align}
\min_{\substack{
\{\boldsymbol{\alpha}, \boldsymbol{\beta}, \boldsymbol{\gamma},\mathbf{x},\\ \{\mathbf{I}^{(r,t)}, \forall r,t\}\}}} &\lim_{T\to \infty} \frac{1}{T} \sum_{t=0}^{T-1}\sum_{z \in \mathcal{Z}}\Big\{T^{\textrm{service}}_{z,t}\textbf{1}_{z,t}+ \lambda W_{z,t}(1-\textbf{1}_{z,t})\Big\} \\
\textrm{s.t.\;\;} & (\ref{c1}),(\ref{c2}),\\
& \sum_{r \in \mathcal{R}} \alpha_{z,r,t} = 1, \sum_{r \in \mathcal{R}} \beta_{z,r,t} = 1, \sum_{r \in \mathcal{R}} \gamma_{z,r,t} = 1 \label{cc1}\\
& \sum_{i = 1}^{N_{r,t}}I^{(r,t)}_{i,j} = 1, \sum_{j = 1}^{N_{r,t}}I^{(r,t)} _{i,j} = 1\label{cc2}\\
&0 \leq x_{z,t}\leq 1,\label{cc3}\\
& \boldsymbol{\alpha}_{z,t}, \boldsymbol{\beta}_{z,t} \in \mathbb{Z}_+^{|\mathcal{R}|}, \label{cc4}\\
& \boldsymbol{I}^{(r,t)}\in \mathbb{Z}_+^{N_{r,t}\times N_{r,t}} \label{cc5},
\end{align}
\end{subequations}
where $\lambda$ represents per-unit penalty,
for the case when the computing offloading service fails.
\begin{figure*}[t]
\scriptsize{
\begin{equation}
\hat{x}_{z} = 
\begin{cases}
\frac{T_{z,h(z)}^\textrm{Q}-\max\{T_{z,r(z)}^\textrm{Q}, T_{z}^\textrm{T}\} + \chi W_z/C_{h(z)}}{\chi W_z/C_{r(z)}+\chi W_z/C_{h(z)}}, \textrm{ if } T_{z,h(z)}^\textrm{Q}-\max\{T_{z,r(z)}^\textrm{Q}, T_{z}^\textrm{T}\} \geq \frac{\chi W_z}{C_{r(z)}}-\chi R_{r(z),h(z)} (T_{z,h(z)}^\textrm{Q}-T_{z}^\textrm{T})(\frac{1 }{C_{r(z)}}+\frac{1 }{C_{h(z)}})\\
\frac{T_{z}^\textrm{T}-\max\{T_{z,r(z)}^\textrm{Q}, T_{z}^\textrm{T}\} + \chi W_z/C_{h(z)} + W_z/R_{r(z),h(z)}}{\chi W_z/C_{r(z)}+\chi W_z/C_{h(z)}+ W_z/R_{r(z),h(z)}},\textrm{ otherwise. }
\end{cases}
\label{eq.tp}
\end{equation}}
\noindent\rule{18cm}{0.4pt}
\end{figure*} 
The optimization variables include three aspects: edge server selection, \textit{i.e.}, $\{\boldsymbol{\alpha}, \boldsymbol{\beta}, \boldsymbol{\gamma}\}$, task partition, \textit{i.e.}, $\mathbf{x}$, and task scheduling, \textit{i.e.}, $\{\mathbf{I}^{(r,t)}, \forall r,t\}$. 
It can be seen that Problem \eqref{problem} is a mixed-integer nonlinear optimization problem. 
Solving the above problem directly by conventional optimization methods is challenging. 
%In our previous work, we provide a model-free solution to deal with the edge server selection problem under a simplified system model. However, applying the solution in \cite{Li2} directly to solve Problem \eqref{problem} is impractical due to the extremely large decision spaces in task scheduling under a 2-D vehicular network. 
{Furthermore, the decision dimension of the problem is too large to apply model-free techniques directly. Taking the variable of task execution order as an example, \textit{i.e.}, $\mathbf{I}^{(r,t)}$, there are $N_{r,t} \times N_{r,t}$ number of decisions to be determined for a server in a time slot. The number of combinations of scheduling decisions is at least $(|\mathcal{Z}|/|\mathcal{R}|)! \times |\mathcal{R}| \times |\mathcal{T}|$, in which tasks are evenly assigned to servers and each task is processed by only one server. Thus, to reduce the decision dimension of the problem, we divide Problem \eqref{problem} into two sub-problems: i) task partition and scheduling problem, and ii) edge server selection problem.} In the task partition and scheduling problem, we aim to obtain the optimal task partition ratio and the execution order to minimize the computing latency given the offloading policy $\{\boldsymbol{\alpha}, \boldsymbol{\beta}\}$. 
After that, we re-formulate the edge server selection problem as an MDP and utilize the DRL technique to obtain the optimal offloading and computing policy.

\section{Task Partition and Scheduling}

Multiple tasks offloaded from different zones can be received by an edge server in a time slot. The computing tasks can only be processed if the tasks scheduled priorly are executed. As a result, the overall computing time may vary depending on the task execution order in edge servers. In addition, the workload of a task can be divided and offloaded to two edge servers, \textit{i.e.}, receiver and helper RSUs. Workload allocation for a task also affects the overall service time. Therefore, we study task partition and scheduling to minimize the service latency given the offloading policy $
\{\boldsymbol{\alpha}, \boldsymbol{\beta}\}$. 
Based on Problem \eqref{problem}, the delay minimization problem can be formulated as follows:
\begin{subequations}
\label{problem2}
\begin{align}
\min_{\mathbf{x}, \{\mathbf{I}^{(r,t)}, \forall r,t\}} &\sum_{z \in \mathcal{Z}}T^{\textrm{service}}_{z,t} \\
\textrm{s.t.\;\;} & (\ref{cc2}),(\ref{cc3}),(\ref{cc5}).
\end{align}
\end{subequations}
Problem \eqref{problem2} is a mixed-integer programming, which involves a continuous variable $\mathbf{x}$ and an integer matrix variable $\{\mathbf{I}^{(r,t)}, \forall r,t\}$. Moreover, even if $\mathbf{x}$ is known, the remaining integer problem is a variation of the traveling salesman problem, which is an NP-hard problem. To reduce the time-complexity in problem-solving, we exploit the properties of task partition and scheduling and develop a heuristic algorithm to obtain an approximate result efficiently.
To simplify the notations, we eliminate the time index $t$ in the remainder of the section since we consider the scheduling scheme for the tasks offloaded in one time slot. 
We further denote $r(z)$ and $h(z)$ as the index of receiver and helper RSUs for task $z$, respectively. 
\begin{lemma}
\label{le.1}
If no task is queued after task $z$ for both the receiver RSU and the helper RSU, the optimal partition ratio for the task $x_z^*$ is $\min\{\max\{0, \hat{x}_z\},1\}$, where $\hat{x}_z$ can be determined by Eq. \eqref{eq.tp}.
\end{lemma}
\begin{proof}
Without considering the tasks queued later, the service time of task $z$ can be minimized by solving the following problem:
\begin{align}
\min & \max \{T^{\textrm{C}}_{z,r(z)}, T^{\textrm{C}}_{z,h(z)}\} & \textrm{  s.t. } & \eqref{cc3}.
\end{align}
Given that $0<x_z<1$, the optimal task partition strategy exists when $T^{\textrm{C}}_{z,r(z)} = T^{\textrm{C}}_{z,h(z)}$. The optimal task partition ratio is $x^*_z = \hat{x}_{z}$. In addition, $x^*_z = \max\{0,\hat{x}_{z}\} = 0$ when the helper RSU can fully process task $z$ in a shorter service time comparing to the queuing time in the receiver RSU, \textit{i.e.}, $\max\{T^{\textrm{Q}}_{z,r(z)},T^{\textrm{T}}_{z}\}\geq \max\{T^{\textrm{Q}}_{z,h(z)},T^{\textrm{T}}_{z}+\frac{\chi W_z}{R_{r(z),h(z)}}\}+\frac{\chi W_z}{C_{h(z)}}$. Otherwise, $x^*_z = \min\{1,\hat{x}_{z}\} = 1$, when the receiver RSU can process task $z$ by itself in a shorter service time comparing to the queuing time in the helper RSU, \textit{i.e.}, $\max\{T^{\textrm{Q}}_{z,r(z)},T^{\textrm{T}}_{z}\}\leq T^{\textrm{Q}}_{z,h(z)}-\frac{\chi W_z}{C_{r(z)}}$.
\end{proof}
Lemma \ref{le.1} shows the optimal partition ratio from the individual task perspective. However, multiple tasks could be offloaded from different zones to an RSU, where the role of the RSU could be different for those tasks. The task partition strategy for a single task could affect the computing latency for the task queued later. Therefore, we will investigate the optimality of the task partition scheme in Lemma \ref{le.1} in terms of minimizing the overall service time for all tasks $z\in \mathcal{Z}$.

\begin{lemma}
Assume that the following conditions are met:
\begin{itemize}
\item The computing capability $C_r$ is identical for all edge servers.
\item The receiver RSU and helper RSU are different for each task, \textit{i.e.}, $r(z) \neq h(z)$.
\item For the helper RSUs for all tasks, the queuing time is not shorter than the offloading time, \textit{i.e.}, $T_{z,h(z)}^Q \geq T^{T}_{z,r(z)}+ T^{R}_{r(z),h(z)}, \forall z,r$.
\end{itemize}
Then, given the execution order of tasks, the optimal solution of Problem \eqref{problem2} follows the results shown in Lemma \ref{le.1}, \textit{i.e.}, $x_z^*=\min\{\max\{0, \hat{x}_z\},1\}, \forall z$.
\label{le.2}
\end{lemma}
\begin{proof}
See Appendix A.
\end{proof}

We have proved that, given the task execution order, the partition ratio in Lemma 1 is the optimal solution for Problem \eqref{problem2} under certain assumptions. Next, we will explore the optimal scheduling order given the workload allocation policy. 

\begin{lemma}
Consider only one available RSU in the system, \textit{i.e.}, $r(z) = h(z)$. Under the assumption in which the offloading time is proportional to the size of the task, the optimal task execution order is to schedule the task with the shortest service time first.
\label{le.3}
\end{lemma}
\begin{proof}
See Appendix B.
\end{proof}
According to the properties provided in Lemmas \ref{le.1}-\ref{le.3}, we design a heuristic algorithm to schedule the task execution order and allocate workload among RSUs. The full algorithm is presented in Algorithm \ref{al.1}. In the algorithm, we allocate the task that has the shortest service time first. For each task, we divide the workload between the receiver RSU and helper RSU according to the optimal partition ratio in Lemma 1. In the worst case, in which all zones have tasks to offload in a time slot, the algorithm requires $|\mathcal{Z}|(|\mathcal{Z}|+1)/2$ iterations to compute the task partition and scheduling results, which can still provide fast responses in the dynamic environment.

\begin{algorithm}[!htb]
  \caption{Task Partition and Scheduling Algorithm (TPSA)}
   \label{al.1}
   {\begin{algorithmic}[1]
\State {At time slot $t$, initialize set $\mathcal{S} = \{z| W_{z,t} \neq 0\}$.}
\State {Initialize $\psi_{r} = T_{r,t}^{\textrm{Q0}}$, $\mathbf{I}^{(r,t)} = \textbf{0}$, and $j_r = 1, \forall r$.}
\While {$|\mathcal{S}| \neq 0$}
\State {Initialize $Q_z = 0, \forall z \in {\mathcal{S}}$.}
\For {Task $z = 1: |\mathcal{S}|$}
\State \multiline{Update $r(z) = \{r|\alpha_{z,r,t} = 1\}$ and $h(z) = \{r|\beta_{z,r,t} = 1\}$.}
\State \multiline{Update partition ratio $x_z = \min\{\max\{0, \hat{x}_z\},1\},$ where $\hat{x}$ is obtained by~\eqref{eq.tp}.}
\State {Update $\hat{\psi}_{z,r(z)} = \psi_{r(z)} + {T^{\textrm{C}}_{z,r(z)}}$.}
\State {Update $\hat{\psi}_{z,h(z)} = \psi_{h(z)} + {T^{\textrm{C}}_{z,h(z)}}$.}
\State {If {$x_z = 1$}, then $Q_{z} = \hat{\psi}_{z,h(z)}$.}
\State {If {$x_z = 0$}, then $Q_{z} = \hat{\psi}_{z,r(z)}$.}
\State {If {$0<x_z <1$}, then $Q_{z} = (\hat{\psi}_{z,r(z)}+\hat{\psi}_{z,h(z)})/2$.}
\EndFor
\State {Find $z^* = \textrm{argmin}_z Q_z$.}
\State {Update $\psi_{r(z^*)} = \hat{\psi}_{z^*,r(z^*)}$ and $\psi_{h(z^*)} = \hat{\psi}_{z^*,h(z^*)}$.}
\State {Update order matrix $I^{r(z^*),t}_{z^*,j_{r(z^*)}} = 1$, and $I^{h(z^*),t}_{z^*,j_{h(z^*)}} = 1$.}
\State {Update $j_{r(z^*)} = j_{r(z^*)}+1$, and $j_{h(z^*)} = j_{h(z^*)}+1$.}
\State {$\mathcal{S} = \mathcal{S}	\backslash \{z^*\}$.}
\EndWhile
\State {$T_{r,t+1}^{\textrm{Q0}} = \psi_{r}-\epsilon, \forall r$.}
    \end{algorithmic}}
\end{algorithm}

\section{AI-Based Collaborative Computing Approach}
\begin{figure*}[t]  
  \centering  
  \includegraphics[width=160mm]{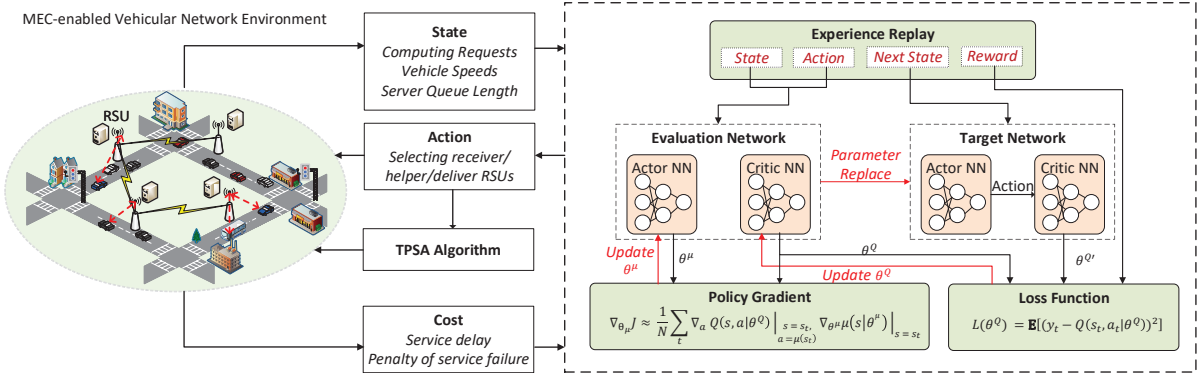}\\
  \caption{AI-based collaborative computing approach.}  
  \label{fig:ddpg1} 
\end{figure*}
To deal with the server selection problem, we utilize a DRL technique to conduct the complex decision-making problem in a dynamic environment. To implement the DRL method, we first re-formulate the problem into an MDP. An MDP can be defined by a tuple $(\mathbb{S}, \mathbb{A}, \mathbb{T}, \mathbb{C})$, where $\mathbb{S}$ represents the set of system states; $\mathbb{A}$ represents the set of actions; $\mathbb{T} = \{p(s_{t+1}|s_t,a_t)\}$ is the set of transition probabilities; and $\mathbb{C}$ is the set of real-value cost functions. The term $C(s,a)$ represents the cost when the system is at state $s\in \mathbb{S}$ and an action $a \in \mathbb{A}$ is taken. A policy $\pi$ represents a mapping from $\mathbb{S}$ to $\mathbb{A}$. 
In our problem, the state space, action space, and cost model in an MDP are summarized as follows:
\begin{enumerate}
\item State space: In time slot $t$, the network state, $s_t$, includes the computing data amount in zones, \textit{i.e.}, $\{W_{z,t}, \forall z\}$, the average vehicle speed, \textit{i.e.}, $\{v_{z,t}, \forall z\}$, and the delay for edge servers to finish the tasks offloaded in previous time slots $\{1, \dots, t-1\}$, \textit{i.e.}, $\{T^{\textrm{Q0}}_{r,t}, \forall r\}$.
\item Action space: For zone $z$ and time slot $t$, the action taken by the network includes three elements: the index of receiver RSU, helper RSU, and deliver RSU, which can be represented by $\{a_{z,t}^1, a_{z,t}^2, a_{z,t}^3\}$, receptively. 
\item Cost model: Given the state-action pair, the overall service time can be available by the TPSA algorithm. Thus, according to the objective function \eqref{problem}, the cost function can be formulated as 
\begin{equation}
C(s_t,a_t) =\sum_{z \in \mathcal{Z}}\Big\{T^{\textrm{service}}_{z,t}\textbf{1}_{z,t}+ \lambda W_{z,t}(1-\textbf{1}_{z,t})\Big\}.
\label{eq.cost_mdp}
\end{equation}  
Then, to obtain the expected long-term discounted cost, the value function $V$ of state $s$ is
\begin{equation}
V(s, \pi) = \mathbb{E}\Big[\sum_{t=0}^\infty \gamma^tC(s_t,a_t)|s_0 = s,\pi \Big],
\end{equation}
where the parameter $\gamma$ is a discount factor. 
%The value function represents the expected costs of all possible state trajectories starting from $s$.
By minimizing the value function of each state, we can obtain the optimal offloading and computing policy $\pi^*$; that is,
\begin{equation}
\pi^*(s) = \textrm{arg\,min}_a \sum_{s'}p(s'|s,a)[C(s,a)+\gamma V(s',\pi^*)].
\label{eq.MDP}
\end{equation}
\end{enumerate}
Due to the limited knowledge on transition probability between the states and the sizeable state-action space in the network, the traditional dynamic programming is not able to find the optimal policy efficiently. Therefore, we adopt DRL to solve the proposed server selection problem. There are three common DRL algorithms: deep Q network (DQN), actor-critic (AC), and DDPG. DQN is a powerful tool to obtain the optimal policy with a high dimension in the state space. Besides an online neural network (evaluation network) to learn the Q value, a frozen network (target network) and the experience replay technique are applied to stabilize the learning process. However, the method shows the inefficiency on the network with a high dimension in the action space, while in our problem, the large number of zones leads the high dimension in both state and action spaces.
On the other hand, both AC and DDPG tackle the problem with a high action dimension by the policy gradient technique. Two networks, \textit{i.e.}, actor and critic networks, are adopted, in which the critic evaluates the Q value, and the actor updates policy parameters in the direction suggested by the critic. Moreover, DDPG combines the characteristics of DQN on top of the AC algorithm to learning the Q value and the deterministic policy by the experience relay and the frozen network, thereby helping reach the fast convergence \cite{ddpg}. In this paper, we exploit the DDPG algorithm to obtain the optimal collaborative computing policy in vehicular networks.

The illustration of our AI-based collaborative computing approach is shown in Fig. \ref{fig:ddpg1}. The system states are observed from the MEC-enabled vehicular network. After state $s_t$ is obtained, the optimal server selection policy can be computed by the DDPG algorithm. According to the server selection results,  the corresponding task partition and scheduling policy can be obtained by the proposed TPSA algorithm. Then, the cost of the corresponding state-action pair and the next system state can be observed from the environment. The state transition set $(s_t,a_t,r_t,s_{t+1})$ is stored in the replay memory for training the neural networks. In DDPG, four neural networks are employed. Two of the four networks are evaluation networks, where the weights are updated when the neural network is trained, and the other two networks are target networks, where the weights are replaced periodically from the evaluation network. For both evaluation and target networks, two neural networks, \textit{i.e.}, actor and critic networks, are adopted to evaluate the optimal policy and Q value, respectively. The weights in evaluation and target critic networks are denoted by $\theta^Q$ and $\theta^{Q'}$, and the weights in evaluation and target actor networks are denoted by $\theta^{\mu}$ and $\theta^{\mu '}$, respectively.

In each training step, a batch of experience tuples are extracted from the experience replay memory, where the number of tuples in a mini-batch is denoted by $N$. The critics in both evaluation and target networks approximate the value function and compute the loss function $L$, where 
\begin{equation}
{L(\theta^Q) = \mathbf{E}\Big[\big(y_t-Q(s_t,a_t|\theta^Q)\big)^2\Big].} \label{ddpg1}
\end{equation}
{The term $Q(s_t,a_t|\theta^Q)$ represents the Q function approximated by the evaluation network.} The value of $y_t$ is obtained from the value function approximated by the target network, where
\begin{equation}
{y_t = C(s_t,a_t)+\gamma Q(s_{t+1},\mu ' (s_{t+1}| \theta^{\mu '})|\theta^{Q'}).}
\label{ddpg2}
\end{equation}
The term $\mu ' (s_{t+1}| \theta^{\mu '})$ represents the action taken at $s_{t+1}$ given by the target actor network. By minimizing the loss function \eqref{ddpg1}, the weights in the evaluation critic, \textit{i.e.}, $\theta^Q$, can be updated. On the other hand, to update the weights of the evaluation actor network, the policy gradient can be represented as   
\begin{align}
{\nabla_{\theta_{\mu}}\! J \! \approx \! \frac{1}{N}\! \sum_t \nabla_a Q(s,a|\theta^Q)|_{\begin{subarray}{l} s = s_t,\\  a = \mu (s_t)\end{subarray}}\nabla_{\theta^{\mu}}\mu (s|\theta^{\mu})|_{s = s_t}.}
\label{ddpg3}
\end{align}
From \eqref{ddpg3}, it can be seen that actor weights are updated in each training step according to the direction suggested by the critic.

\begin{figure*}[t]  
  \centering  
  \includegraphics[width=120mm]{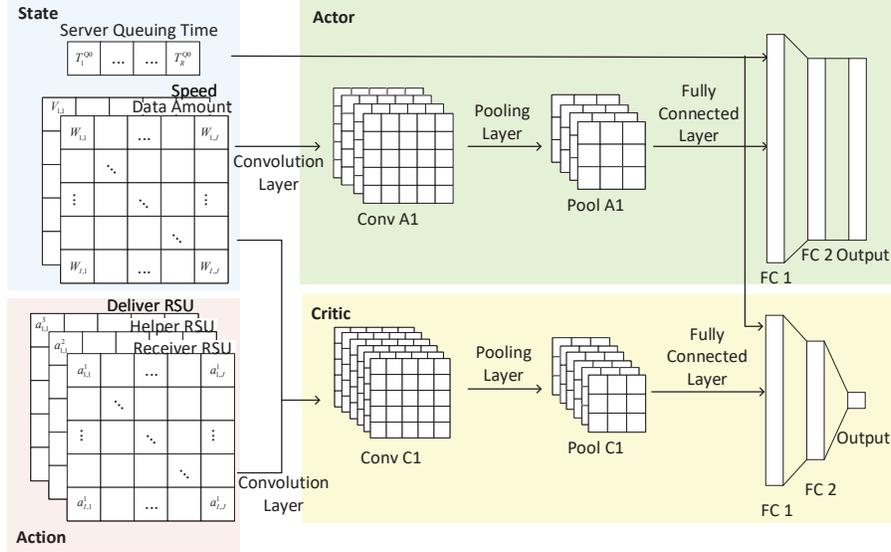}\\
  \caption{The structure of actor and critic neural network. }  
  \label{fig:ddpg2} 
\end{figure*}

Although the DDPG algorithm is able to tackle the problem with a high dimension of state and action spaces, it is inefficient to apply the DDPG algorithm directly in our problem due to the 2-dimensional transportation network and the multiple dimensions of the input.  A huge number of neurons in a network will be deployed if the conventional neural network with fully connected layers is adopted. To improve the algorithm efficiency, we utilize CNN in both actor and critic networks to exploit the correlation of states and actions among different zones. The structure of actor and critic networks is shown in Fig. \ref{fig:ddpg2}. Before fully connected layers, convolution layers and pooling layers are applied to learn the relevant features of the inputs among zones. 
Due to the weight sharing feature of CNN filters, the number of training parameters can be significantly reduced compared to the network with fully connected layers \cite{CNN}. After several convolution and pooling layers, the output of the CNN combines the state of edge servers and forwards to fully connected layers.

The proposed AI-based collaborative computing approach is provided in Algorithm \ref{al}, where $\tau$ is a small number less than 1.
In our algorithm, to learn the environment efficiently, the system will continuously train the parameter by $N_t$ times after $N_e$ time step, where $N_e>N_t$.

\begin{algorithm}[!htb]
  \caption{AI-based Collaborative Computing Approach}
   \label{al}
   {\begin{algorithmic}[1]
     \State{Initialize critic network $Q(s_0,a_0|\theta^{Q})$ and actor network $\mu(s_0|\theta^{\mu})$ with weights $\theta^{Q}$ and $\theta^{\mu}$. }
     \State{Initialize target network with weights $\theta^{Q'} = \theta^{Q}$ and $\theta^{\mu '} = \theta^{\mu}$. }
     \State {Initialize the experience replay buffer.}
     \State \multiline{Initialize a random vector $\mathcal{N}$ as the noise for action
exploration.}
     \For{ episode = 1:G}
\State{Initialize environment, and observe the initial state $s_0$.}
\For{ time slot $t = 1:T$}
\State{Select action $a_t = \mu (s|\theta^{\mu})+\mathcal{N}$.}
\State \multiline{Let $\alpha_{z, a^1_{z,t},t}$, $\beta_{z, a^2_{z,t},t}$, and $\gamma_{z, a^3_{z,t},t}$ equal to 1.} 
\State \multiline{Compute the task partition and scheduling results by Algorithm 1.}
\State{Observe next state $s_{t+1}$ and cost $C(s_t,a_t)$.}
\State \multiline{Store transition $(s_t,a_t,r_t,s_{t+1})$ into the experience replay buffer. Delete the oldest transition set if the buffer is full.}
\If{ $k \mod N_e == 0$}
\For{ $j = 1:N_t$}
\State \multiline{Sample a mini-batch of $N$ samples. }
\State{Update $y_t$ by \eqref{ddpg2}.}
\State \multiline{Update the weights in the evaluation critic network by minimizing the loss in \eqref{ddpg1}.}
\State \multiline{Update the weights in the evaluation actor network using sampled policy gradient presented in \eqref{ddpg3}.}
\State \multiline{Update target networks: $\theta^{Q'} = \tau \theta^{Q}+(1-\tau) \theta^{Q'}$; $\theta^{\mu '} = \tau \theta^{\mu}+(1-\tau) \theta^{\mu '}$.}
\EndFor
\EndIf
\EndFor
\EndFor
    \end{algorithmic}}
\end{algorithm}

\section{Performance Evaluation}
{In this section, we first present the efficiency of the proposed TPSA algorithm in task partition and scheduling. Then, we evaluate the performance of the proposed AI-based collaborative computing approach in a vehicular network simulated by VISSIM \cite{vissim}, where TPSA is applied to schedule computing tasks according to the policy given by the DDPG algorithm. 
}
\subsection{Task Partition and Scheduling Algorithm}
We first evaluate the performance of the proposed TPSA algorithm. In the simulation, we consider that tasks can be offloaded to five edge servers with an identical offloading rate of 6 Mbits/s. The communication rate among the servers is 8 Mbits/s. We set that the computing capability of the servers is 8 GC/s, and the number of computation cycles needed for processing 1 Mbit is 4 GC. The computing data amount of tasks is uniformly distributed in the range of [1,21] Mbits. For each task, the receiver and helper RSUs are randomly selected from the five servers. We compare the proposed TPSA algorithm with \textit{brute-force} and \textit{random} schemes. In the brute-force scheme, we utilize an exhaustive search for finding the optimal scheduling order. In the random scheme, we randomly assign the scheduling order of the tasks. Note that, for both \textit{brute-force} and \textit{random} schemes, we adopt the optimal task partition ratio in workload allocation. The simulation results presented in Figs. \ref{fig:sim2}(a) and \ref{fig:sim2}(b) are averaged over 200 rounds of Monte Carlo simulations.

The service delay performance of the proposed algorithm is shown in Fig. \ref{fig:sim2}(a). It can be seen that an increase in the task number leads to increasing overall service time, and the increasing rate of the random scheme is the highest among the three schemes. The proposed TPSA algorithm can achieve a performance very close to the brute-force scheme. Moreover, we compare the runtime between the proposed TPSA and the brute-force scheme. As shown in Fig.~\ref{fig:sim2}(b), as the number of the task increases, the runtime of brute-force scheme increases exponentially, 
%When the number of tasks is large, the brute-force scheme consumes significant computing time to search an optimal scheduling order, 
while the proposed TPSA algorithm has imperceptible runtime to compute the scheduling result that is close to the optimal one. In summary, the proposed TPSA algorithm can achieve a near-optimal performance for task partition and scheduling with low computation complexity.
\begin{figure}[t]
\centering  
  \begin{subfigure}
  \centering
  \includegraphics[width=80mm]{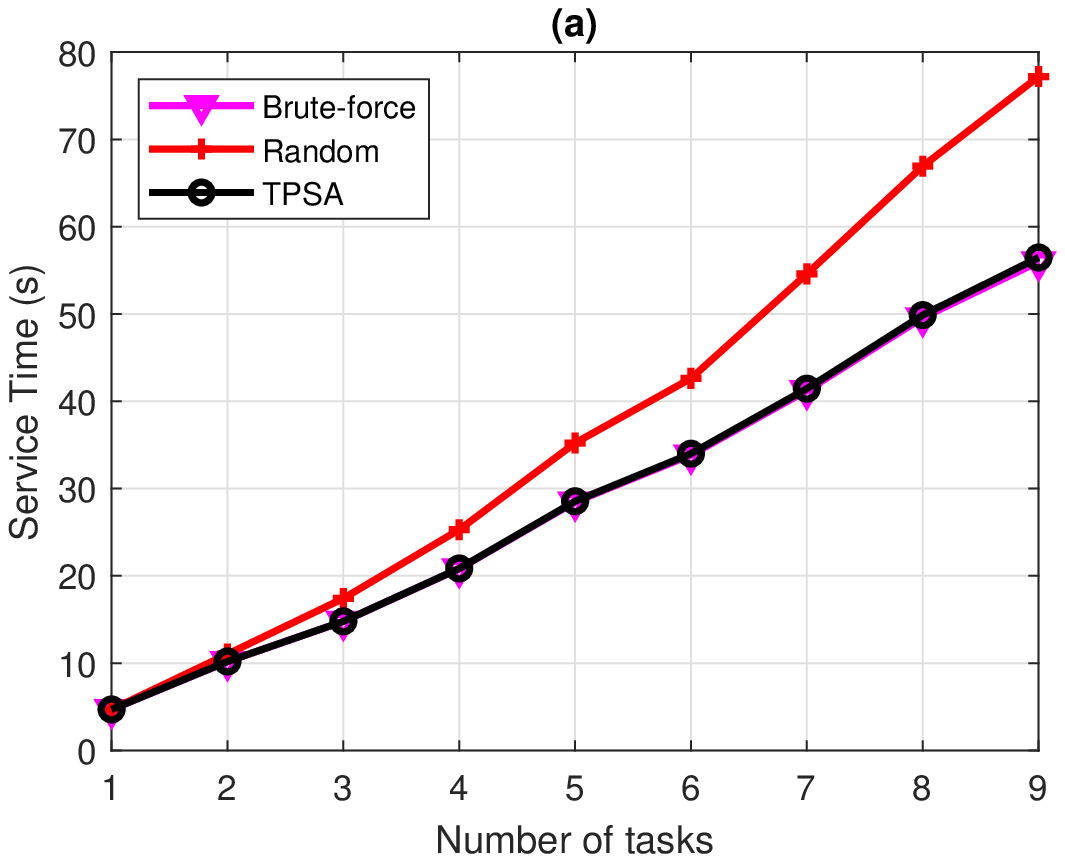}
  \end{subfigure}
  \hspace{2cm}
\begin{subfigure}
  \centering  
  \includegraphics[width=80mm]{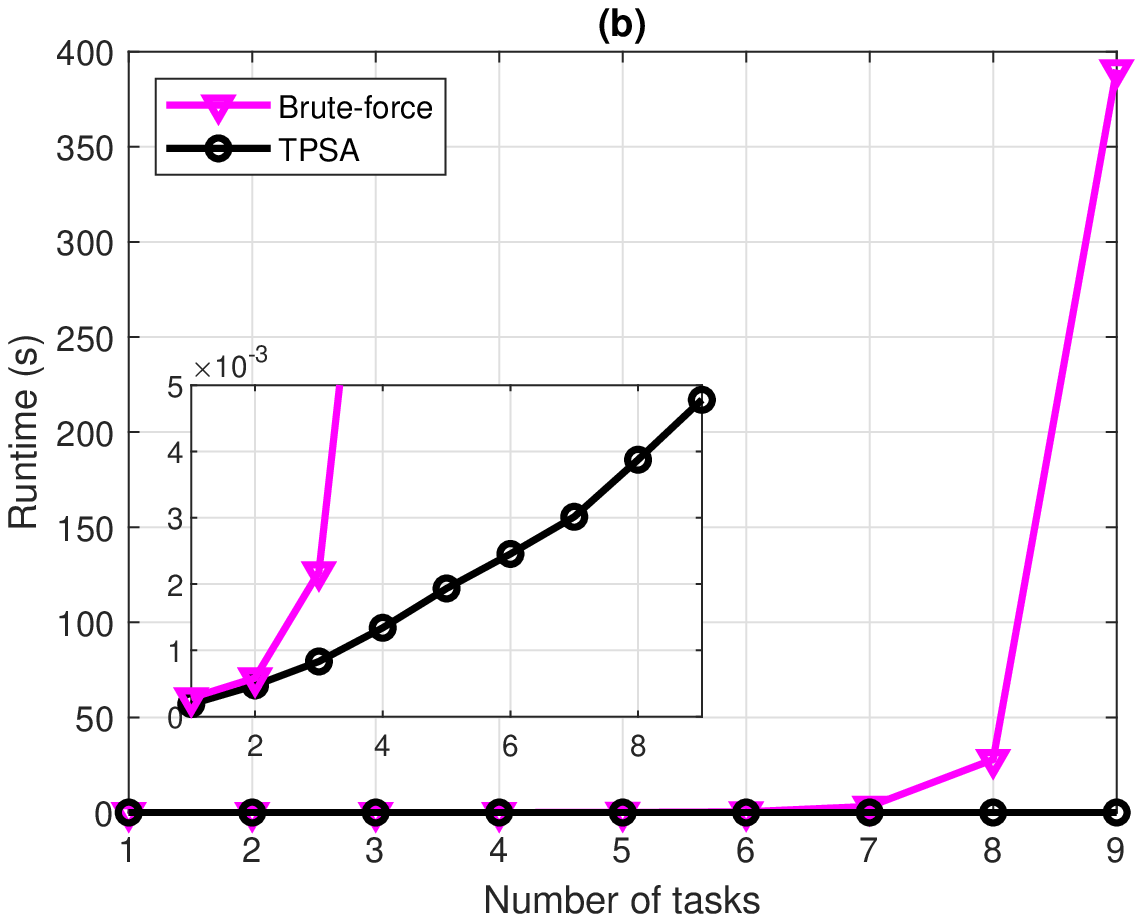}
  \end{subfigure}
  \caption{(a) Average service delay among the three task partition and scheduling schemes with respect to the number of tasks. (b) Average computation runtime among the three task partition and scheduling schemes with respect to the number of tasks.} \label{fig:sim2} 
\end{figure}

\subsection{AI-based Collaborative Computing Approach}
\begin{table}[t]
\centering
\caption{Network Parameters}
\label{t1}
\scriptsize{
\begin{tabular}{|c|c|c|c|c|}
\hline
$P^V$    & $P^R$       & $\sigma_r^2$, $\sigma_v^2$ & $\lambda$                & $\epsilon$                  \\ \hline
27 dBm   & 37 dBm      & -93 dBm    & 50                       & 1 s                         \\ \hhline{|=|=|=|=|=|}
$f$      & $\chi$      & $N_e, N_t$ & $\delta^{\textrm{O}}$ & $\delta^{\textrm{D}}$ \\ \hline
2800 MHz & 1200 C/bits & 80, 25     & 7 dB                     & 7 dB                       \\ \hline
\end{tabular}}
\end{table}

\begin{figure}[t]  
  \centering  
  \includegraphics[width=50mm]{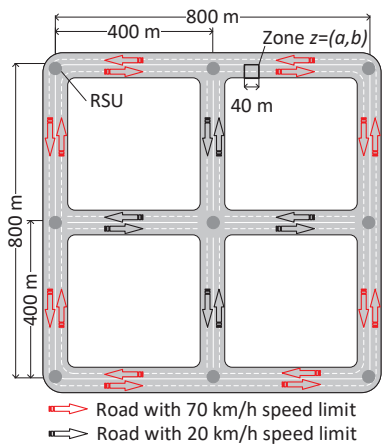}\\
  \caption{The transportation network topology for simulation.}  
  \label{fig:sim3} 
\end{figure}

In this subsection, we evaluate the performance of the proposed AI-based collaborative computing approach. In the simulation, we consider an 800 m $\times$ 800 m transportation system, where the transportation topology is shown in Fig. \ref{fig:sim3}. Nine RSUs with edge servers are deployed, as indicated in the figure. We generate vehicle traffic by VISSIM~\cite{vissim}, where 200 vehicles are traveling in the area. The speed of vehicles depends on the speed limit on the road and the distance to the vehicle ahead. For each vehicle, the computing tasks are generated using a Poisson process, and the input data amount of each task is uniformly distributed in the range of [2,5] Mbits.  The length and width of a zone are 40 m and 10 m (2 driving lanes), respectively. Other network parameter settings are presented in Table \ref{t1}. We test the system performance within a duration of 20 seconds.

\begin{table}[t]
\centering
\caption{Neural Network Structure}
\label{t2}
\scriptsize{
\begin{tabular}{|l|l|l|}
\hline
\multicolumn{3}{|c|}{Actor Network}                          \\ \hline
Layer     & Number of neurons       & Activation function    \\ \hline
CONV1     & 5$\times$1$\times$2$\times$10, stride 1      & relu                   \\ \hline
POOL1     & 2$\times$1                     & none                   \\ \hline
\multicolumn{3}{|l|}{Data Concatenation and Batch Normalization Layer} \\ \hline
FC1       & 1400                    & tanh                   \\ \hline
FC2       & 1400                    & tanh                   \\ \hline
FC3       & 5$\times A\times B$                   & tanh                   \\ \hline
\multicolumn{3}{|c|}{Critic Network}                         \\ \hline
Layer     & Number of neurons       & Activation function    \\ \hline
CONV1     & 5$\times$1$\times$2$\times$40, stride 1      & relu                   \\ \hline
POOL1     & 2$\times$1                     & none                   \\ \hline
CONV2     & 3$\times$1$\times$40$\times$10, stride 1     & relu                   \\ \hline
POOL2     & 2$\times$1                     & none                   \\ \hline
\multicolumn{3}{|l|}{Data Concatenation and Batch Normalization Layer} \\ \hline
FC1       & 640                     & relu                   \\ \hline
FC2       & 512                     & relu                   \\ \hline
FC3       & 128                     & none                   \\ \hline
FC4       & 1                       & relu                   \\ \hline
\end{tabular}}
\end{table}
The neural network structures of the DDPG algorithm are presented in Table \ref{t2}. The initial learning rates of the actor and critic networks are 1e-5 and 1e-4, respectively, and the learning rates are attenuated by 0.991 in every 500 training steps. The experience replay buffer can adopt 8,000 state-action transitions, and in each training step, the number of transition tuples selected for training, \textit{i.e.}, the batch size, is 128. We adopt a soft parameter replacement technique to update the parameters in the target network, where $\tau$ is 0.01. We compare the performance of the proposed AI-based collaborative computing approach with three approaches. In the \textit{Greedy} approach, vehicles always offload their tasks to the RSU with the highest SNR, and the received computing tasks will not be collaboratively computed with other RSUs. In the \textit{Greedy+TPSA} approach, a vehicle offload their tasks to the RSU with the highest SNR, and the RSU randomly selects another RSU to compute the task collaboratively. The task partition and scheduling policy follows the TPSA algorithm, and the computing results are delivered by the receiver RSU. In the \textit{Random+TPSA} approach, the receiver, helper, and deliver RSUs are selected randomly, and the TPSA algorithm is applied to determine the task partition ratio and the execution order.

\begin{figure}[t]  
  \centering  
  \includegraphics[width=90mm]{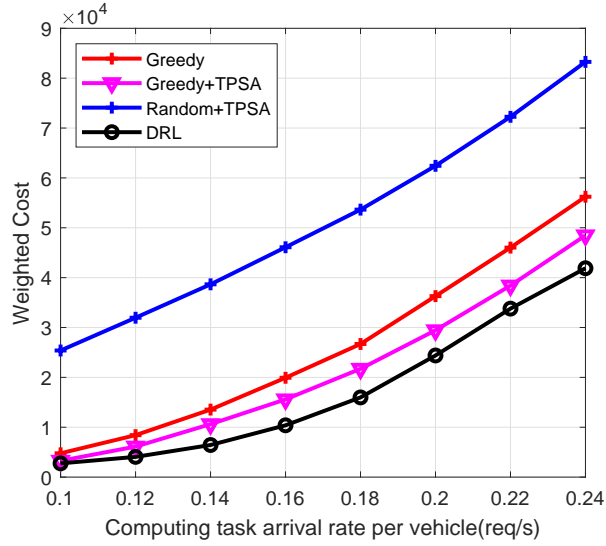}\\
  \caption{Average weighted computing delay cost versus computing task arrive
rate per vehicle.}  
  \label{fig:sim4} 
\end{figure}

The overall weighted computing cost with respect to task arrival rates is shown in Fig. \ref{fig:sim4}. Our proposed approach can achieve the lowest computing cost compared to the other three approaches. The random approach suffers the highest cost compared to others due to the inefficient server selection in the scheme. Moreover, the greedy+TPSA approach achieves a lower cost compared to the greedy approach. The reason is that parallel computing is able to reduce the overall service time, and the proposed TPSA is able to achieve near-optimal task partition and scheduling results. However, the greedy approach selects the servers according to the instantaneous cost of the network rather than the value in the long term. Therefore, the greedy+TPSA approach cannot attain a lower cost compared to the proposed AI-based approach.

\begin{figure}[t]  
  \centering  
  \includegraphics[width=90mm]{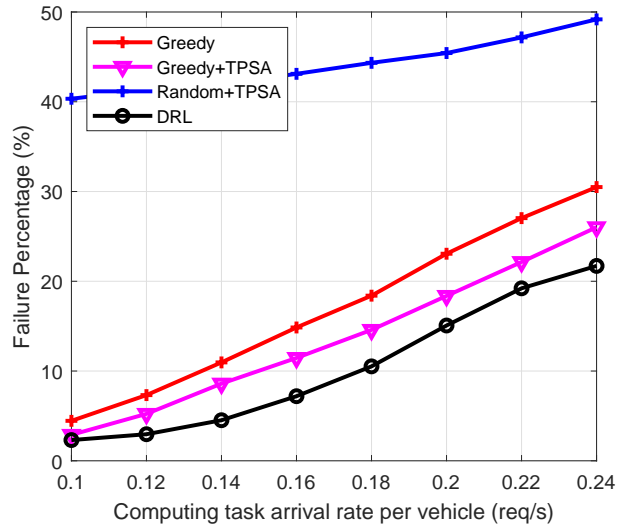}\\
  \caption{Average percentage of service failure versus computing task arrive rate per vehicle.}  
  \label{fig:sim5} 
\end{figure}
\begin{figure}[t]  
  \centering  
  \includegraphics[width=90mm]{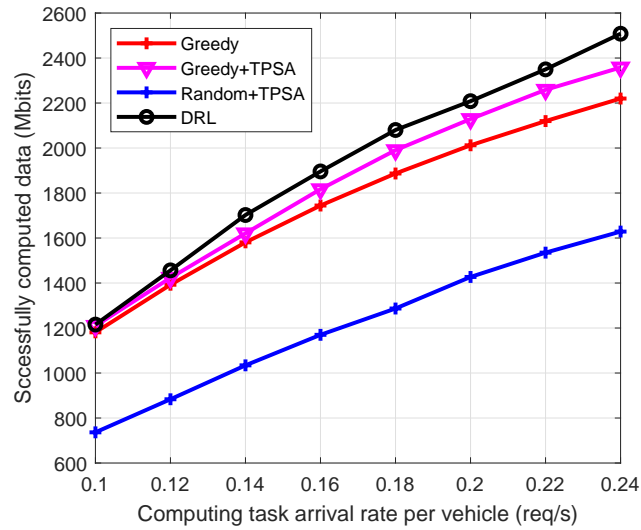}\\
  \caption{Average computing data amount which is successfully computed versus
computing task arrive rate per vehicle.}  
  \label{fig:sim6} 
\end{figure}

As indicated in Eq. \eqref{eq.cost_mdp}, the service cost consists of the service delay and the failure penalty. The results of the service failure percentage is shown in Fig. \ref{fig:sim5}. Similar to the service cost, the proposed AI-based approach achieves the lowest failure percentage among the four approaches. Correspondingly, as shown in Fig. \ref{fig:sim6}, the proposed approach can successfully process the highest amount of data among the four approaches. On the other hand, the results of the average service delay for 1 Mbits successful computed data are shown in Fig. \ref{fig:sim7}. Compared to the other three approaches, the proposed scheme reduces the service delay significantly. Furthermore, the delay of the random approach increases exponentially since less amount of data can be successfully computed when the task arrival rate is high. 

\begin{figure}[t]  
  \centering  
  \includegraphics[width=90mm]{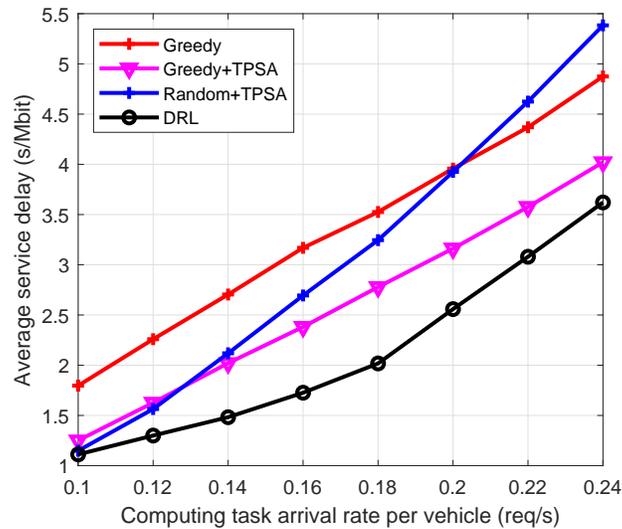}\\
  \caption{Average service delay for 1 Mbits successful computed data versus
computing task arrive rate per vehicle.}  
  \label{fig:sim7} 
\end{figure}

The convergence performance of the proposed AI-based approach is shown in Fig. \ref{fig:sim8}, where the highlighted line represents the moving average from 50 samples around the corresponding point. Note that in our algorithm, we explore multiple times in each training step. It can be seen that our approach converges after 10,000 episodes, or equivalently, after the network being trained by around 3,000 episodes, \textit{i.e.}, 60,000 training steps.
\begin{figure}[t]  
  \centering  
  \includegraphics[width=90mm]{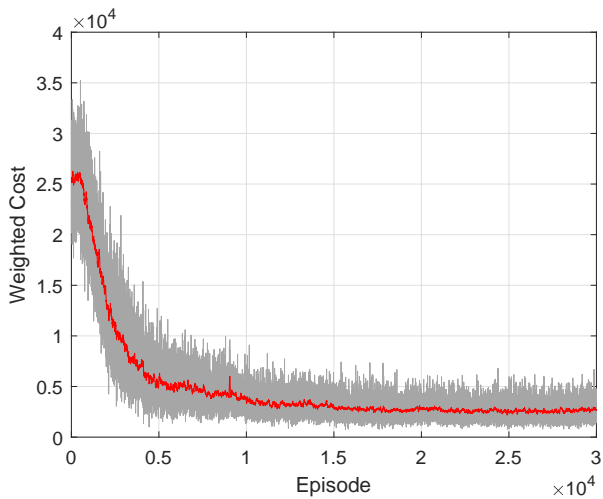}\\
  \caption{Convergence performance of the proposed algorithm, where the task arrival rate is 0.1 request/sec.}  
  \label{fig:sim8} 
\end{figure}

\section{Conclusion}
We have introduced a novel collaboration computing framework to reduce computing service latency and improve service reliability in MEC-enabled vehicular networks. 
The proposed framework addresses the challenge of maintaining computing service continuity for vehicle users with high mobility. As a result, our collaborative computing approach is able to support proactive decision making for computation offloading through learning the network dynamics.
%, the dynamics in vehicular networks can be learned, and 
%An AI-based approach has been developed to learn the network dynamics, and 
%Specifically, we first develop an efficient task partition and scheduling scheme to make the complex cooperative computing problem scalable. We further develop an AI-based collaborative computing approach to select the optimal edge servers for task offloading, computing, and result delivering. By learning network dynamics, the proposed approach supports proactive decision making on computation offloading for high-speed users. 
Our work can be applied to offer low-latency and high-reliable edge computing services to vehicle users in a complex network environment, such as urban transportation systems.
In the future, we will investigate multi-agent learning approach to compute the optimal computing strategy under the limited information collected by the edge servers.

%The first contribution of this work is the development of an efficient task partition and scheduling scheme to minimize the overall computing time for the tasks offloaded to edge servers. We further develop an AI-based collaborative computing approach to select the optimal edge servers for task offloading, computing, and result delivering. A CNN based DDPG technique is employed to tackle the complex and highly dynamic network environment. Simulation results validate the effectiveness of the proposed approach. 

\section*{Appendix A: Proof of Lemma~\ref{le.2}}
\begin{figure}[t]  
  \centering  
  \includegraphics[width=80mm]{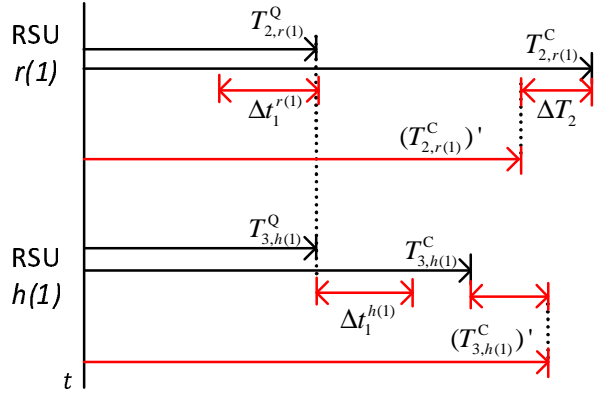}\\
  \caption{An illustration of task partition. }  
  \label{fig:app} 
\end{figure}
An illustration of task partition is shown in Fig. \ref{fig:app}. Consider that the workload of all tasks are divided and shared among RSUs following the results in Lemma \ref{le.1}. We focus on a single task which is numbered as task 1 as shown in the figure.  As indicated in the second and the third assumptions in Lemma 2, the computing load of task 1 is shared between RSU $r(1)$ and $h(1)$. Tasks 2 and 3 are scheduled after task 1 in RSU $r(1)$ and $h(1)$, respectively. In addition, $T_{2,h(2)}^Q \geq T^{T}_{2,r(2)}+ T^{R}_{r(2),h(2)}$.
We then prove that, under the assumption in Lemma \ref{le.2}, the overall service time will be increased if the partition ratio of task 1 does not follow the policy presented in Lemma 1. 

Consider that, for task 1, the workload assigned to RSU $r(1)$ is decreased by $\Delta x$. Correspondingly, the computing time of task 1 in server $r(1)$ is reduced by $\Delta t_1^{r(1)} = \Delta x/C_{r(1)}$, while the computing time of task 1 in server $h(1)$ is increased by $\Delta t_1^{h(1)} = \Delta x/C_{h(1)}$. Thus, the service time of task 1 is increased by $ \Delta T_1 = \Delta x/C_{h(1)} $. Denote the new partition ratio of task 2, after task partition ratio $x_1$ is decreased by $\Delta x$, as $\hat{x}_2$.
We then list following cases to analyze the time deduction from the tasks queued after the task 1:
\begin{itemize}
\item \textit{Case 1:} Task 2 regards RSU $r(1)$ as the receiver RSU, \textit{i.e.}, $r(1) = r(2)$, and $\hat{x}_2 < 1$. 
According to Eq. \eqref{eq.service} and Lemma 1, the optimal service time of task 2 is
\begin{align}
&T^{\textrm{service}}_2 = \max\{T^{\textrm{T}}_2, T^{\textrm{Q}}_{2,r(1)}\} \notag\\
&\hspace{1cm}+ \frac{(T^{\textrm{Q}}_{2,h(2)}-\max\{T^{\textrm{T}}_2, T^{\textrm{Q}}_{2,r(1)}\})C_{h(2)}+\chi W_2}{C_{r(1)}+C_{h(2)}}.
\label{eq.le2.1}
\end{align}
After task partition ratio $x_1$ is decreased by $\Delta x$, task 2 can be processed by RSU $r(1)$ in advance by $\Delta t_1^{r(1)}$. The new optimal service time of task 2 is 
\begin{align}
&(T^{\textrm{service}}_2)' = \max\{T^{\textrm{T}}_2, T^{\textrm{Q}}_{2,r(1)}-\Delta t_1^{r(1)}\}\!\notag\\
&\hspace{0.2cm} +\!\frac{(T^{\textrm{Q}}_{2,h(2)}-\max\{T^{\textrm{T}}_2, T^{\textrm{Q}}_{2,r(1)}-\Delta t_1^{r(1)}\})C_{h(2)}+\chi W_2}{C_{r(1)}+C_{h(2)}}.
\label{eq.le2.2}
\end{align}
The service time deduction on task 2 can be obtained by subtracting Eq. \eqref{eq.le2.1} by Eq. \eqref{eq.le2.2}. We found the reduced service time $\Delta T_2 \leq \Delta t_1^{r(1)}C_{r(1)}/(C_{r(1)}+C_{h(2)})$, where equality can be reached when $T^{\textrm{T}}_2\leq T^{\textrm{Q}}_{2,r(1)}-\Delta t_1^{r(1)}$.

\item \textit{Case 2:} Task 2 regards RSU $r(1)$ as the receiver RSU, \textit{i.e.}, $r(1) = r(2)$, and $\hat{x}_2 = 1$. In this case, the new optimal service time of task 2 is  
\begin{align}
&(T^{\textrm{service}}_2)' = \max\{T^{\textrm{T}}_2, T^{\textrm{Q}}_{2,r(1)}-\Delta t_1^{r(1)}\}+ \frac{\chi W_2}{C_{r(1)}}.
\label{eq.le2.3}
\end{align}
Via subtracting Eq. \eqref{eq.le2.1} by Eq. \eqref{eq.le2.3}, we have 
\begin{align}
\Delta T_2 &\leq \Delta t_1^{r(1)} - \frac{(\chi W_z/C_{r(1)}-T^{\textrm{Q}}_{2,h(2)}+T^{\textrm{Q}}_{2,r(1)})C_{h(2)}}{C_{r(1)}+C_{h(2)}} \notag\\
&\leq \frac{\Delta t_1^{r(1)}C_{r(1)}}{(C_{r(1)}+C_{h(2)})},
\end{align}
where equality can be achieved when $T^{\textrm{T}}_2\leq T^{\textrm{Q}}_{2,r(1)}-\Delta t_1^{r(1)}$.

\item \textit{Case 3:} Task 2 regards RSU $r(1)$ as the helper RSU, \textit{i.e.}, $r(1) = h(2)$. In this case, the new optimal service time of task 2 is  
\begin{align}
&(T^{\textrm{service}}_2)' = \max\{T^{\textrm{T}}_2, T^{\textrm{Q}}_{2,r(2)}-\Delta t_1^{r(2)}\} \notag\\
&\hspace{0.1cm}+\frac{(T^{\textrm{Q}}_{2,r(1)}-\Delta t_1^{r(2)}-\max\{T^{\textrm{T}}_2, T^{\textrm{Q}}_{2,r(2)}\})C_{r(1)}+\chi W_2}{C_{r(2)}+C_{r(1)}}.
\end{align}
Similar as case 1, the reduced service time for task 2 is $\Delta T_2 = \Delta t_1^{r(1)}C_{r(1)}/(C_{r(1)}+C_{r(2)})$.
\end{itemize}
Considering that the computing capabilities $C_{r}$ are identical for all servers (the first assumption in Lemma 2), the maximum service time deduction for task 2 is $\Delta t_1^{r(1)}/2$. For all tasks queued after task 1 in RSU $r(1)$, the overall service time deduction is less than $\Delta t_1^{r(1)}[1/2 + (1/2)^2+ (1/2)^3+\dots]$, which is always less than $\Delta t_1^{r(1)}$. We omit the proof for the case when the workload assigned in RSU $h(1)$ is decreased by $\Delta x$ due to the similarity. 
Therefore, we obtain that, under the assumptions presented in Lemma 2, the overall service time will be increased if the workload allocation does not follow the task partition ratio presented in Lemma 1.

\section*{Appendix B: Proof of Lemma~\ref{le.3}}

Suppose the tasks in edge server $r$ are scheduled by the shortest-task-first rule, and task 2 is queued after the task 1. Then, we have
\begin{align}
\max\{T_{1,r}^\textrm{Q},T_{1,r}^\textrm{T}\} +T_{1,r}^\textrm{P} \leq \max\{T_{1,r}^\textrm{Q},T_{2,r}^\textrm{T}\} +T_{2,r}^\textrm{P}. 
\label{eq.condition1}
\end{align}
If the order of task 1 and task 2 are switched with each other, the service time of task 2 will be decreased by 
\begin{equation}
D = \max\{\max\{T_{1,r}^\textrm{Q},T_{1,r}^\textrm{T}\} +T_{1,r}^\textrm{P},T_{2,r}^\textrm{T}\}-\max\{T_{1,r}^\textrm{Q},T_{2,r}^\textrm{T}\}.
\end{equation}
On the other hand, the service time of task 1 will be increased by 
\begin{equation}
I = \max\{T_{1,r}^\textrm{Q},T_{2,r}^\textrm{T}\} +T_{2,r}^\textrm{P}- \max\{T_{1,r}^\textrm{Q},T_{1,r}^\textrm{T}\}.
\end{equation}
From \eqref{eq.condition1}, we can derive that $I \geq T_{1,r}^\textrm{P}$.
Then, the overall service time of tasks 1 and 2 will be increased by
\begin{align}
&I-D \geq  T_{1,r}^\textrm{P} - \max\{\max\{T_{1,r}^\textrm{Q},T_{1,r}^\textrm{T}\} +T_{1,r}^\textrm{P},T_{2,r}^\textrm{T}\}\notag\\
&\hspace{5cm}+\max\{T_{1,r}^\textrm{Q},T_{2,r}^\textrm{T}\}.
\end{align}
We then list the three scenarios on $T_{2,r}^\textrm{T}$:
\begin{itemize}
\item \textit{Case 1:} $T_{2,r}^\textrm{T} \geq \max\{T_{1,r}^\textrm{Q},T_{1,r}^\textrm{T}\} +T_{1,r}^\textrm{P}$. In this case, $I-D \geq T_{1,r}^\textrm{P} \geq 0$.
\item \textit{Case 2:} $T_{1,r}^\textrm{Q}\leq T_{2,r}^\textrm{T} \leq \max\{T_{1,r}^\textrm{Q},T_{1,r}^\textrm{T}\} +T_{1,r}^\textrm{P}$. In this case, $I-D \geq  T_{2,r}^\textrm{T} - \max \{T_{1,r}^\textrm{Q},T_{1,r}^\textrm{T}\}$. According the assumption, where $T_{1,r}^\textrm{T} \leq T_{2,r}^\textrm{T}$,  we then have $I-D \geq 0$.
\item \textit{Case 3:} $T_{2,r}^\textrm{T} \leq T_{1,r}^\textrm{Q}$. In this case, $I-D \geq T_{1,r}^\textrm{Q} - \max \{T_{1,r}^\textrm{Q},T_{1,r}^\textrm{T}\} = 0$.
\end{itemize}
Therefore, we can obtain the conclusion that the service time will be increased if the task execution order does not follow a shortest-task-first rule under the assumptions. 
 
{\footnotesize
\bibliographystyle{IEEEbib}
\bibliography{reference}
}

\end{document}